\RequirePackage{fix-cm}
\documentclass[smallextended]{svjour3}       % onecolumn (second format)
\smartqed  % flush right qed marks, e.g. at end of proof
\usepackage[utf8]{inputenc}
\usepackage[T1]{fontenc}
\usepackage[scaled=0.92]{helvet}  
\usepackage[noadjust]{cite}

\usepackage{graphicx}
\usepackage{array}
\newcolumntype{M}[1]{>{\centering\arraybackslash}m{#1}}
\usepackage{tabularx,booktabs}
\newcolumntype{C}{>{\centering\arraybackslash}X} % centered version of "X" type
\usepackage{dblfloatfix}
    
\usepackage{amsmath}
\usepackage{dsfont}
\usepackage{hhline,booktabs}
\usepackage{tablefootnote}
\usepackage{enumitem}
\usepackage[T1]{fontenc}
\usepackage{adjustbox}
\usepackage{booktabs, multirow}
\usepackage{amsmath}

\usepackage[scaled=0.9]{DejaVuSansMono}
\usepackage{listings}
\usepackage{bm}
\usepackage{url}
\usepackage{mathtools}
\usepackage{balance}
\usepackage{subcaption}
\allowdisplaybreaks
\usepackage{algorithm}
\usepackage{algpseudocode}
\usepackage{etoolbox}
\usepackage[mathscr]{euscript}
\usepackage{calc}
\usepackage{pgfplots}
\usepackage{tikz}
\usepackage{setspace}
\usepackage{pgfplotstable}
\usepackage[font=small]{caption}
\usepackage{eqparbox,array}
\usepackage{stmaryrd}

\DeclareMathAlphabet{\pazocal}{OMS}{zplm}{m}{n}
\DeclareSymbolFont{missing}{OML}{cmr}{m}{n}
\DeclareMathSymbol{\ell}{\mathord}{missing}{'140}

\usetikzlibrary{calc}
\usetikzlibrary{patterns}
\usetikzlibrary{spy,backgrounds}
\pgfplotsset{grid style={dotted,gray}}

\DeclareMathOperator*{\argmax}{arg\,max}

\newcommand{\maximize}{%
  \mathopen{}\operatorname*{maximize}%
}

\newcommand{\subjto}{\textup{subject to}}

\makeatletter
\tikzset{reset label anchor/.code={%
    \let\tikz@auto@anchor=\pgfutil@empty
    \def\tikz@anchor{#1}
  },
  reset label anchor/.default=center
}

\usetikzlibrary{decorations.pathmorphing,arrows.meta,positioning}
\tikzstyle{printersafe}=[decoration={amplitude=0pt}]

\makeatother

\usetikzlibrary{automata}
\usetikzlibrary{calc,fit,shapes.geometric}
\pgfdeclarelayer{signal}
\pgfsetlayers{signal,main}
\tikzstyle{printersafe}=[segment amplitude=0 pt]
\usetikzlibrary{arrows}% To get more arrow heads

\usetikzlibrary{calc}
\usetikzlibrary{decorations.pathreplacing,decorations.markings,shapes.geometric}

\newsavebox{\mybox}
\pgfplotsset{compat=1.16}

\begin{document}

%\bstctlcite{IEEEexample:BSTcontrol}

\titlerunning{Multi-Agent DRL for V2X}

\title{Network Slicing for Vehicular Communications: A Multi-Agent Deep Reinforcement Learning Approach}

\author{Zoubeir~Mlika \and Soumaya~Cherkaoui}%
\institute{Z. Mlika \at
              Department of Electrical and Computer Engineering\\
              Université de Sherbrooke, Canada \\
              \email{zoubeir.mlika@usherbrooke.ca}%
           \and
           S. Cherkaoui \at
              Department of Electrical and Computer Engineering\\
              Université de Sherbrooke, Canada \\
              \email{soumaya.cherkaoui@usherbrooke.ca}%
}
\maketitle

\begin{abstract}
  This paper studies the multi-agent resource allocation problem in vehicular networks using non-orthogonal multiple access (NOMA) and network slicing. Vehicles want to broadcast multiple packets with heterogeneous quality-of-service (QoS) requirements, such as safety-related packets (e.g., accident reports) that require very low latency communication, while raw sensor data sharing (e.g., high-definition map sharing) requires high-speed communication. To ensure heterogeneous service requirements for different packets, we propose a network slicing architecture. We focus on a non-cellular network scenario where vehicles communicate by the broadcast approach via the direct device-to-device interface (i.e., sidelink communication). In such a vehicular network, resource allocation among vehicles is very difficult, mainly due to (i) the rapid variation of wireless channels among highly mobile vehicles and (ii) the lack of a central coordination point. Thus, the possibility of acquiring instantaneous channel state information to perform centralized resource allocation is precluded. The resource allocation problem considered is therefore very complex. It includes not only the usual spectrum and power allocation, but also coverage selection (which target vehicles to broadcast to) and packet selection (which network slice to use). This problem must be solved jointly since selected packets can be overlaid using NOMA and therefore spectrum and power must be carefully allocated for better vehicle coverage. To do so, we first provide a mathematical programming formulation and a thorough NP-hardness analysis of the problem. Then, we model it as a multi-agent Markov decision process. Finally, to solve it efficiently, we use a deep reinforcement learning (DRL) approach and specifically propose a deep Q learning (DQL) algorithm. The proposed DQL algorithm is practical because it can be implemented in an online and distributed manner. It is based on a cooperative learning strategy in which all agents perceive a common reward and thus learn cooperatively and distributively to improve the resource allocation solution through offline training. We show that our approach is robust and efficient when faced with different variations of the network parameters and compared to centralized benchmarks.
\end{abstract}


\section{Introduction}\label{sec:intro}
Vehicle-to-everything (V2X) communication is an integral part of future intelligent transportation systems (ITS) and has become a key communication paradigm in future wireless networks~\cite{triwinarko2021phy}. Indeed, V2X communication is one of the key verticals in fifth generation (5G) networks~\cite{8459911,doi:10.1002/9781119514848.ch9}. In 3rd generation partnership project (3GPP) release 16, the new radio cellular-V2X (NR C-V2X) communication supports two transmission modes: a direct mode with NR C-V2X sidelink, and an over-the-cellular network mode, that together, provide seamless connectivity to vehicles through a unified radio. V2X communication enables information sharing through vehicle-infrastructure (V2I), vehicle-to-pedestrian, vehicle-to-vehicle (V2V), or vehicle-to-cellular network (V2N) communication. V2X communication covers a range of use cases~\cite{alalewi20215g} that can broadly be divided into safety services, non-safety services and infotainment services. Safety services are the main use case of V2X communication, and aim to improve road safety by enabling road users and/or infrastructure to exchange information in a timely manner to avoid accidents. Non-safety services are offered by ITS to improve traffic management and efficiency. Infotainment services include other value-added services for vehicle users such as video streaming or content sharing~\cite{8723326,doi:10.1002/ett.3652}. Despite the convenience of a unified radio, the diverse service requirements of vehicular applications make the problem of resource allocation in V2X communication very challenging~\cite{7564925,7511574,7577072,7585028}. A key challenge is to support future 5G vehicular services with extremely diverse performance requirements using a unified radio interface. Network slicing is a potential solution to solve this key challenge~\cite{8459911,doi:10.1002/ett.3652}. Network slicing belongs to software defined networking~\cite{8004158} and is a tool that allows network operators to support virtualized end-to-end networks by creating different logical networks on top of a common, programmable physical infrastructure.

    In this paper, we propose a novel solution to solve the challenging problem of resource allocation in a 5G vehicular network using network slicing. Specifically, we focus on V2V sidelink communication without the assistance of a base station such as a 5G-NR next generation NodeB (gNB). Network slicing in such scenarios is more challenging because the proposed solutions must be distributed and implemented independently in each vehicle. In V2V communication, 5G-NR has defined different ways to send information, including broadcast, unicast and groupcast communication techniques~\cite{3gpp.38.885}. In this article, the broadcast technique is considered. On the other hand, non-orthogonal multiple access (NOMA) is a promising technique for increasing spectrum efficiency in wireless networks in general~\cite{8357810} and in vehicular networks in particular~\cite{8246842}. NOMA has been shown to be effective in providing broadband communications and massive connectivity because it allows multiple transmitting devices to share common spectrum resources (in time and/or frequency)~\cite{7974737}.

    In this paper, we propose network slicing-based resource allocation framework using NOMA for V2V broadcast communication. We answer the following important question: ``how to distributively allocate the different resources in network slicing-based vehicular network in order to guarantee the latency, reliability and rate requirements of the different slices?''.

    The problem considered, hereafter called vehicle resource allocation (VRA), is a four-dimensional problem that involves the allocation of four resources corresponding to the following four decisions: deciding which packet to send, deciding the broadcast range, deciding which resource blocks (RBs) to use for transmission, and deciding the transmission power. The use of NOMA is attractive for the development of distributed solutions because some RBs can be reused by multiple transmitting vehicles and thus there is no need to pay attention to collisions. By carefully allocating different vehicle resources (slice, coverage, RBs, and power) and applying successive interference cancellation (SIC) to the corresponding receivers, NOMA can help increase the capacity of various vehicle network applications. It should be noted that the use of NOMA in broadcast communications is different from the usual NOMA technique in wireless uplink and downlink networks and is more difficult to apply. This is due to the nature of broadcasting in vehicle networks: two transmitting vehicles broadcasting with two different transmission powers to the \textit{same} group of receiving vehicles must carefully distribute their transmission powers so that the corresponding receivers can successfully apply SIC.

\subsection{Contributions}
    To the best of our knowledge, this is the first work that distributively solves VRA in 5G NR C-V2X sidelink communication based on network slicing and NOMA. To do so, we apply deep reinforcement learning (DRL)~\cite{9318243,8943940,abouaomar2021deep}. In general, deep learning (DL) has had significant success in various disciplines and is applied differently to various problems in vehicular networks~\cite{8345672} and reinforcement learning~\cite{9383093,abouaomar2021service} helps in solving complex problems. With the help of DRL, we provide a distributed and model-free solution to VRA. On the other hand, we propose a model-based formulation using integer programming and analyze the NP-hardness of VRA. The contributions of our work are summarized in the following list.
\begin{itemize}
    \item We formulate VRA as a mixed integer nonlinear program.
    \item We prove that VRA is NP-hard even in its simplest form.
    \item We model VRA as a multi-agent Markov decision process. 
    \item We propose a deep Q learning (DQL) algorithm to provide an online and distributed resource allocation solution that incites, through proper reward function design, the agents (the vehicles) to work cooperatively.
    \item We compare DQL to optimal and centralized benchmark algorithms that we have adapted to solve VRA and we demonstrate the superior performance of DQL.
\end{itemize}

\subsection{Related Work}

    Di \textit{et al.}~\cite{7974737} used NOMA to solve the RB and power allocation problem in a V2X communication network where the objective is to maximize the number of successfully decoded signals subject to scheduling and fairness constraints. The proposed solution is applied in two phases: a centralized phase and a distributed phase. In the centralized phase, the authors solved the RB allocation problem by transforming it into an assignment problem and propose a rotation-based algorithm. In the distributed phase, the authors proposed a distributed algorithm to find the transmission power of vehicles to improve the performance of NOMA. Liang \textit{et al.}~\cite{7922547} focused on maximizing the sum rate of V2I links while guaranteeing minimum quality of service (QoS) for V2I and V2V links. The minimum QoS of V2V links ensured reliable communication by imposing a minimum threshold value for the outage probability of the achievable rate. The authors assumed that the channel state information (CSI) at the base station is delayed due to high mobility. The authors decoupled the problem by considering the case of a single V2I user and a single V2V user and derived the power allocation vector. Then, using the allocated power, they found the spectrum allocation by transforming the problem into an assignment problem and used the Hungarian algorithm to solve it. The proposed solution is based on a central coordination point and does not include network slicing or vehicle coverage optimization. Nasir \textit{et al.}~\cite{8792117} solved the classical power allocation problem in an interference channel using model-free DRL. Particularly, the proposed solution can be summarized as follows. After collecting CSI and QoS information, each transmitter adapts its own transmit power accordingly. The objective is, as usual, to maximize the total weighted sum-rate. The proposed method used DQL to develop a distributed solution. Liang \textit{et al.}~\cite{8792382} proposed a multi-agent DRL framework to solve the challenging problem of resource allocation in a V2X communication network. In particular, the problem involved spectrum sharing between V2I links and V2V links and the objective was to ensure ultra-low latency communication for V2V links and high throughput for V2I links. The proposed solution is based on DQL to provide a distributed solution for the RB and power allocation problem. Campolo \textit{et al.}~\cite{10.1145/3331054.3331549} used the LTE V2X Mode 4 scenario to improve the decoding probability of V2V links while considering full duplex radios. The overall objective was to improve the existing approach used in LTE V2X Mode 4 and specifically the work aims to improve resource allocation and reduce collisions and packet decoding. However, the resource allocation should be updated every period based on the collision history which may increase the overhead. Liang \textit{et al.}~\cite{7913583} examined resource allocation in vehicular device-to-device communications. In particular, the authors solved the problem of spectrum and power allocation while maximizing the throughput of V2I links and ensuring the reliability of V2V links. Zhang \textit{et al.}~\cite{10.1007/978-3-030-22971-9_38} studied a vehicle network using multi-access edge computing (MEC) with backup cloud servers. They proposed a DQL algorithm to offload vehicle tasks to the MEC servers to minimize the total processing time of vehicle tasks. The results showed superiority in terms of processing time, especially when backup cloud servers are considered. Sharif \textit{et al.}~\cite{sharif2021a} studied the problem of clustering vehicles in Internet of vehicles (IoV) networks. They proposed a DRL-based approach to select cluster heads for resource allocation among vehicles. Their approach is based on the DRL actor-critic method where the policy and the value function are separated. The policy is called the actor because it selects actions while the value function is called the critic because it criticizes the actions chosen by the actor. Their approach was compared to static and DQN-based approaches and found to outperform them, especially in terms of convergence time, i.e., the actor-critic approach required fewer iterations to achieve better throughput. Zhang \textit{et al.}~\cite{8944302} studied the problem of mode selection (V2V mode or V2I mode) and resource allocation (RB and power) in V2X communication networks. The problem is formulated to maximize the total capacity of V2I users and guarantee the latency and reliability of V2V users. The authors proposed a multi-agent DRL approach in a two-time scale architecture. In a large time scale, a graph-theoretic based vehicle clustering algorithm is proposed. In a small time scale, vehicles in the same group cooperate to form a single DRL model based on federated learning. Ye \textit{et al.}~\cite{8633948} considered RB and power allocation in V2X communication networks. The aim was to design a multi-agent DRL algorithm where V2V agents learn to reduce interference with all V2I links and with other V2V links while meeting their latency requirements. A DQL algorithm is proposed with low communication overhead.%Ren \textit{et al.}~\cite{a1} discussed vehicle clustering approaches to improve the network scalability and the connection reliability. Different clustering techniques were explored including cluster head selection, cluster formation, and cluster maintenance.

    In the above previous work, the integration of network slicing and NOMA in V2V communication networks with the consideration of broadcast communication in 5G-NR C-V2X assisted sidelink communication scenario is not well analyzed and solved. The challenging multidimensional problem of coverage, slice, RB, and power optimization is not solved distributively and without the help of cellular network in previous work. In this work, we fill this research gap by first proposing a mathematical programming model and then studying its NP-hardness. Then, we use DRL to propose a distributed multi-agent solution to this multidimentional problem.

\subsection{Organization}
    The paper is organized as follows. Section~\ref{sec:mod} presents the model including the system assumptions and the mathematical programming model and it studies the NP-hardness of the problem. Section~\ref{sec:sol} presents the multi-agent deep reinforcement learning framework and describes the proposed algorithm. Section~\ref{sec:sim} presents centralized benchmark algorithmic solutions and gives some simulation results. Finally, section~\ref{sec:cl} draws some conclusions.

\section{Model}\label{sec:mod}
\subsection{System Model}
    We consider a vehicular network composed of $n+m$ vehicles. A set $\mathscr{V}$ of $m$ vehicles generates traffic and is called transmitting vehicles. The remaining vehicles form a set $\mathscr{W}$ of $n$ receiving vehicles. All vehicles operate in the 5G-NR C-V2X-assisted sidelink communication scenario~\cite{3gpp.38.885} and thus communicate with each others via the direct device-to-device interface. The transmitting vehicles transmit according to the broadcast communication approach~\cite{9088326}. Since transmitting vehicles are not able to broadcast to all other vehicles, coverage selection must be optimized. In other words, a transmitting vehicle $v\in\mathscr{V}$ can broadcast only to a subset $\mathscr{W}_v\subseteq\mathscr{W}$ of the receiving vehicles. The proposed resource allocation is carried out autonomously by the vehicles according to the autonomous mode of 5G-NR~\cite{9088326}. The total time duration is divided into a set $\mathscr{T}$ of $T$ time-slots of length $\tau$ seconds each. The total transmission bandwidth is divided into a set $\mathscr{F}$ of $F$ frequency-slots. A resource block (RB) is given by the pair $(f,t)$ for each frequency-slot $f$ and time-slot $t$~\cite{3gpp.36.211}.

    Several use cases can be supported by our model, including cooperative communication using extended sensors~\cite{9088326}. To be able to provide guaranteed heterogeneous service requirements of these use cases, we propose a communication architecture based on network slicing. Network slicing is an effective solution to satisfy the requirements of various use cases of wireless networks in general and vehicular networks in particular~\cite{doi:10.1002/ett.3652}. This is done primarily by creating logical networks on top of a common, programmable infrastructure~\cite{doi:10.1002/ett.3652}. It is known that network slicing involves the core network (CN) as well as the radio access network (RAN). In this paper, we consider network slicing in the RAN only. In our system model, we create two slices. The first slice is reserved for non-safety applications and is called the non-safety slice (or slice 1). It is designed primarily to support non-safety related traffic that is characterized by high-throughput transmissions, e.g., video streaming. On the other hand, the second slice is reserved for safety applications and is referred to as the safety slice (or slice 2), which is designed primarily for safety-related traffic that is characterized by extremely low latency requirements~\cite{a2}, e.g., emergency alerts and accident reports. An example of our system model is given in Fig.~\ref{sysmod1}.
\begin{figure}[ht!]
  \centering
  \includegraphics[scale=0.4]{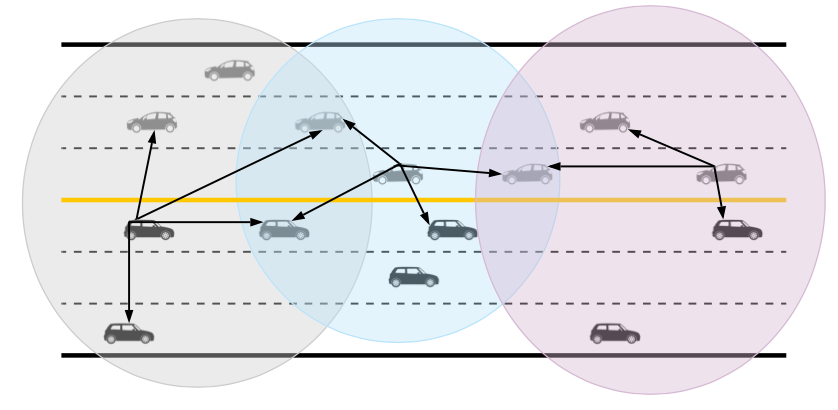}
\caption{An example of the system model with three transmitting vehicles.}
\label{sysmod1}
\end{figure}

Each transmitting vehicle $v\in\mathscr{V}$ generates two packets of size [in bits] $\sigma_{v}^\textsf{n}$ and $\sigma_{v}^\textsf{s}$. Without loss of generality, the packet sizes are called packet requirements. Each packet corresponds to a specific traffic (either non-safety related traffic or safety related traffic). Based on our network slicing architecture, a non-safety related packet is transmitted on the logical non-safety network slice and requires a high throughput. On the contrary, a safety-related packet is transmitted on the logical safety network slice and thus has strict latency requirements. We call these two packets ``non-safety packet'' and ``safety packet'' and refer to them as $\textsf{pkt}_{v}^\textsf{n}$ and $\textsf{pkt}_{v}^\textsf{s}$ respectively. For a packet to be successfully transmitted, the corresponding transmitting vehicle must be assigned a number of RBs such that the number of achievable bits on those assigned RBs is at least equal to the requirements of the corresponding packet. Note that the non-safety packet $\textsf{pkt}_{v}^\textsf{n}$ corresponding to the vehicle $v$ can be assigned any of $\mathscr{F}\times\mathscr{T}$ RBs so that its requirements are met. On the other hand, the safety packet $\textsf{pkt}_{v}^\textsf{s}$ arrives during the considered time horizon at a specific time-slot $s_{v}$. Due to the strict latency requirements of $\textsf{pkt}_{v}^\textsf{s}$, we associate with it a positive integer $e_v$, where $s_v\leq e_v\leq T$, that denotes its deadline. Thus, $\textsf{pkt}_{v}^\textsf{s}$ must meet its requirements by scheduling it on any frequency-slot $f$ such that for the RB $(f,t)$, the time-slot $t$ belongs to $\{s_v.. e_v\}$\footnote{The notation $\{a..b\}$ with $b\geq a$ denotes the integer interval between $a$ and $b$ (inclusive).}, i.e., the admissible set of RBs belongs to $\mathscr{F}\times\{s_v..e_v\}$.

The normalized power gain of the wireless channel between $v\in\mathscr{V}$ and $w\in\mathscr{W}$ on the RB $(f,t)$ is given by $g_{v,w,f}(t)=|h_{v,w,f}(t)|^2/N_0$, where $N_0$ is the noise power and $h_{v,w,f}(t)$ is the channel coefficient that includes fast and slow fading. Each $v\in\mathscr{V}$ transmits $\textsf{pkt}_{v}^i$ (for $i\in\{\textsf{s},\textsf{n}\}$) with a transmission power $p_{v,f}^i(t)$ through RB $(f,t)$ that belongs to $[0,p_v^{\textsf{max}}]$.

In the considered system model, the selection of vehicle coverage, i.e., the selection of the receiving vehicles targeted for broadcast, is a significant challenge. According to~\cite{3gpp.37.885}, a transmitting vehicle transmits to all intended receiving vehicles within a certain distance, and this distance must be carefully designed to improve the performance of vehicular networks. The selection of vehicle coverage can be seen as vehicle clustering that was recently discussed in~\cite{a1}. Thus, for successful packet transmission, a transmitting vehicle $v\in\mathscr{V}$ must: (i) decide whether to send both packets or only the safety packet, (ii) decide which receiving vehicles to broadcast to, (iii) allocate the necessary RBs, and (iv) allocate the transmission powers $p_{v,f}^\textsf{s}(t)$ and $p_{v,f}^\textsf{n}(t)$ for each allocated RB $(f,t)$. To improve the spectral efficiency of the system, non-orthogonal multiple access (NOMA) can be used to overlay transmissions from the transmitting vehicles to a common receiving vehicle. The common receiving vehicle then uses successive interference cancellation (SIC) to decode the superimposed transmissions. For simplicity, if a transmitting vehicle transmits both packets, it uses orthogonal RBs, i.e., NOMA is applied in different transmitting vehicles and not in the transmission of the same transmitting vehicle.

\subsection{VRA Optimization Model}
Let $x_{v,w,f}^\textsf{n}(t)=1$ (resp. $x_{v,w,f}^\textsf{s}(t)=1$) if and only if vehicle $v$ transmits to vehicle $w$ the non-safety packet (resp. the safety packet) on the RB $(f,t)$ and $y_{v,w}^\textsf{n}=1$ (resp. $y_{v,w}^\textsf{s}=1$) if and only if vehicle $v$ transmits the non-safety packet (resp. the safety packet) to vehicle $w$. 

The signal-to-interference-plus-noise ratio (SINR) of the $i$th packet ($i$ denotes either the non-safety packet or the safety packet, i.e., $i\in\textsf{n},\textsf{s}\}$) between vehicle $v\in\mathscr{V}$ and vehicle $w\in\mathscr{W}_v$ on the RB $(f,t)$ is given by :
\begin{align}
  \label{sinrraw}
  \textsf{sinr}_{v,w,f}^{i}(t)\coloneqq\dfrac{p_{v,f}^i(t)g_{v,w,f}(t)}{1+I_{v,w,f}(t)},
\end{align}
where $I_{v,w,f}(t)$ is the interference generated by other transmitting vehicles broadcasting on RB $(f,t)$, which is given by:
\begin{align}\label{ni}
I_{v,w,f}(t)=\sum\limits_{v'\in\mathscr{V}_{v,w,f}(t)}g_{v',w,f}(t)\bigl(p_{v',f}^\textsf{n}(t)+p_{v',f}^\textsf{s}(t)\bigr),
\end{align} 
and $\mathscr{V}_{v,w,f}(t)=\{v'\in\mathscr{V}: g_{v,w,f}(t) > g_{v',w,f}(t)\}$. In~\eqref{ni}, we used $p_{v',f}^\textsf{n}(t)+p_{v',f}^\textsf{s}(t)$ because we assumed that a transmitting vehicle sends its packets on orthogonal RBs.

Our main focus is on optimizing the packet reception ratio (PRR) in V2X communication networks. Technically, it is called Type 2 PRR and is defined in~\cite{3gpp.37.885} as follows: for a packet and a transmitting vehicle, the PRR is given by the percentage of vehicles with successful reception among the total number of receiving vehicles. For this purpose, we optimize the total number of vehicles with successful reception. Therefore, the objective function is given by :
\begin{align}\label{objfunc}
    \sum_{v\in\mathscr{V}}\sum_{w\in\mathscr{W}_v}\bigl(y_{v,w}^\textsf{n}+y_{v,w}^\textsf{s}\bigr).
\end{align}
We note that the objective function is an unweighted sum of the number of successfully delivered packets, which means that safety and non-safety packets are treated the same. Nevertheless, since safety packets are generally lighter in bit size, they should be delivered more often than non-safety packets as we show in the simulation results. Adding the weights $\lambda_\textsf{s}$ and $\lambda_\textsf{n}$ to the objective function is a straightforward approach to prioritize each type of packet, but further analysis of fairness is needed, which will be left for our future work.

The vehicular resource allocation (VRA) problem is formulated as an optimization program as follows. 
\begin{subequations}
\label{pb:1}
  \begin{align}%{2}
  \maximize &\quad \sum_{v\in\mathscr{V}}\sum_{w\in\mathscr{W}_v}\bigl(y_{v,w}^\textsf{n}+y_{v,w}^\textsf{s}\bigr),\label{pb:1a}\\
  \subjto
  & \quad x_{v,w,f}^\textsf{n}(t),x_{v,w,f}^\textsf{s}(t),y_{v,w}^\textsf{n},y_{v,w}^\textsf{s}\in\{0,1\},\forall v,w,f,t,\label{pb:1b}\\
  & p_{v,f}^i(t)\in[0,p_v^{\textsf{max}}],\forall v,f,t,\label{pb:1c}\\
  & \sum_{f\in\mathscr{F}}\sum_{t\in\mathscr{T}}\beta\tau\lg\bigl(1+\textsf{sinr}_{v,w,f}^i(t)\bigr)\geq\sigma_v^iy_{v,w}^i,\forall v,w,i,\label{pb:1d}\\
  & x_{v,w,f}^i(t)+x_{v,w',f'}^i(t)\leq1,\forall v,t,f\ne f',w,w',i,\label{pb:1e}\\
  & x_{v,w,f}^i(t)\leq y_{v,w}^i,\forall v,w,f,t,i,\label{pb:1f}\\
  & y_{v,w}^i\leq\sum_{f\in\mathscr{F}}\sum_{t\in\mathscr{T}}x_{v,w,f}^i(t),\forall v,w,i,\label{pb:1g}\\
  & x_{v,w,f}^\textsf{n}(t)+x_{v,w,f}^\textsf{s}(t)\leq1,\forall v,w,f,t,\label{pb:1h}\\
  & x_{v,w,f}^\textsf{s}(t)=0,\forall v,w,f,t\notin\{s_v..e_v\},\label{pb:1i}\\
  & p_{v,f}^i(t)\leq p_v^{\textsf{max}}x_{v,w,f}^i(t),\forall v,w,f,t,i,\label{pb:1j}
  \end{align}
\end{subequations}
% \clearpage
The objective function and the constraints are explained in the following list.
\begin{itemize}
\item The objective function in~\eqref{pb:1a} represents the number of successfully received packets. 
\item Constraints~\eqref{pb:1b} and~\eqref{pb:1c} represent the optimization variables. 
\item Constraints~\eqref{pb:1d} force a packet to be successfully received only if its achievable bits (where $\beta$ is the bandwidth and $\tau$ is the time-slot duration) on the allocated RBs are greater than the minimum required bits. 
\item For the sake of fairness, constraints~\eqref{pb:1e} require that transmitting vehicle $v$ at time-slot $t$ should not use more than one frequency-slot to send packet $i$. 
\item Constraints~\eqref{pb:1f} and~\eqref{pb:1g} relate the variables $x_{v,w,f}^i(t)$ and $y_{v,w}^i$.
\item Constraints~\eqref{pb:1h} state that transmitting vehicle $v$ should use orthogonal RBs to send both the safety and the non-safety packets.
\item Constraints~\eqref{pb:1i} guarantee that vehicle $v$ should transmit its safety packet only in the range $\{s_v..e_v\}$ of time-slots.
\item Constraints~\eqref{pb:1j} guarantee that the transmission power of vehicle $v$ should be zero if $v$ does not transmit any packet.
\end{itemize}

Problem~\eqref{pb:1} is nonlinear and nonconvex because of the constraints~\eqref{pb:1c}. We show in the next subsection that it is in addition NP-hard even for a special case that can be formulated as an integer linear program. Thus,~\eqref{pb:1} is very challenging to solve in an optimal way and it is even more difficult to solve it in a distributed manner without assuming the existence of a central coordination point.

\subsection{NP-hardness}

\begin{lemma}\label{lemma1}
    VRA is NP-hard.
\end{lemma}

\begin{proof}
    We prove the lemma by restriction. We consider the following restricted version of VRA:
    \begin{itemize}
    \item[1.] there is only the safety slice;
    \item[2.] there is a single frequency-slot;
    \item[3.] there are $m$ transmitting vehicles and a single receiving vehicle $w$;
    \item[4.] the arrival and deadline of the safety packet of vehicle $v$ are $s_v=1$ and $e_v=T$, respectively;
    \item[5.] OMA technique is assumed; and
    \item[6.] all transmitting vehicles are allocated their maximum power.
    \end{itemize} 

	Under this restriction, VRA becomes equivalent to the following problem: maximize the number of safety packets successfully received by $w$ while allocating time-slots to the transmitting vehicles such that (i) the minimum required rate of vehicle $v$ is met, and (ii) the transmitting vehicles do not use the same time-slots.

    We prove NP-hardness of this restricted version of VRA by reducing the maximum independent set (MIS) problem~\cite{Garey:1979} to it in polynomial time. 

    Let an instance of MIS be given by a graph. Vertices represent transmitting vehicles while edges represent time-slots. An edge exists between two vertices if and only if one of the corresponding vehicles can be scheduled in the corresponding time-slot. We can construct the channel coefficients between transmitting vehicle $v$ and $w$ at time-slot $t$ as follows: if the edge corresponding to $t$ is incident to the vertex corresponding to vehicle $v$, then the channel coefficient is set to $1$, otherwise it is set to $0$. The transmission power and the noise power are set to $1$. The minimum number of bits required by vehicle $v$, $\sigma_v^{textsf{s}}$, is chosen equal to the degree of vertex $v$. 
    
    If there exists an independent set $\mathscr{IS}$ in the given graph of maximum size, then by scheduling vehicle $v\in\mathscr{IS}$ at time-slot $t$ (for each $t$ incident to $v$), we have a maximum number of scheduled vehicles. Since we have an independent set, it is true that each time-slot is used by at most one vehicle. Moreover, by construction of the channel coefficients, the rate reached by vehicle $v$ is equal to the degree of vertex $v$. On the other hand, if there is a solution to the restricted version of VRA of maximum size, then, to meet the minimum number of bits required, vehicle $v$ must be scheduled in all time-slots incident to it. The set of scheduled vehicles forms an independent set of maximum size since each time-slot is used by at most one vehicle. 

    In summary, MIS reduces to the restricted version of VRA in polynomial time and the latter is solved if and only if the former is solved. The proof of the lemma follows since MIS is a well known NP-hard problem~\cite{Garey:1979}.
\end{proof}

\section{Multi-Agent Deep Reinforcement Learning Based Resource Allocation}\label{sec:sol}

When vehicles operate in 5G-NR C-V2X assisted sidelink scenario, the resource allocation problem becomes very difficult to solve due to the lack of a central coordination central point acting as a network manager. Therefore, it is necessary for vehicles to solve the resource allocation problem by themselves. To this end, we use multi-agent DRL to develop an online distributed algorithm to solve VRA. The proposed algorithm is based on the well-known deep Q learning (DQL) approach~\cite{mnih-dqn-2015}. DQL operates in two phases: the learning (or the training) phase and the inference (or the testing) phase. In the training phase, each agent trains a deep neural network, often referred to as deep Q network (DQN). In the inference phase, each agent, based on its observations, takes actions based on its trained DQN.

DQL is a major improvement of the table-based method called Q-learning. Q-learning works by creating a table of state-action pairs and finding the best action given a certain state. It often uses an exploration policy called the $\epsilon$-greedy method in which an action is chosen at random with probability $\epsilon$ and is chosen to give the best reward so far otherwise. Q-learning can solve an interesting set of RL problems. However, when the state and action spaces become large (which is often the case in resource allocation problems in wireless networks), creating a Q-table and finding the best policy becomes a prohibitively complex task. In addition, many states will be very rarely visited. Moreover, the task of learning tabular Q becomes even more complex when the learning is multi-agent. 

DQL is a promising approach that can be used to solve the curse of dimensionality in RL~\cite{mnih-dqn-2015} by approximating the Q function instead of using a Q table. One way to solve the multi-agent problem in RL is to combine DQN with independent Q learning for each agent. In other words, each agent tries to learn its own policy based on its own observations and actions while treating all other agents as part of the environment. However, this will strongly influence the outcome of the training phase as it will create a non-stationary environment. Therefore, we will also discuss how to remove the non-stationarity problem by creating a specific state space.

Before describing the proposed DQL algorithm, we first model VRA as a multi-agent Markov decision process (MDP) defined by a state space, an action space, a reward function, and a probability transition function. Each transmitting vehicle, as an independent agent in a given state, performs an action, receives a reward, and transitions to the next state based on its interaction with the vehicle environment. This interaction with the unknown vehicle environment allows each agent to gain experiences and increase its accumulated rewards. The interaction of the agents with the environment is illustrated in Fig.~\ref{fig:mdp}. To maximize system performance, i.e., the objective function of VRA in~\eqref{pb:1}, the agents must act cooperatively. By specifically designing the reward function, we are able to construct a multi-agent RL framework in which agents cooperatively maximize the objective function of VRA.
\tikzset{%
  block/.style    = {draw, thick, rectangle, minimum height = 3em,
    minimum width = 3em},
}
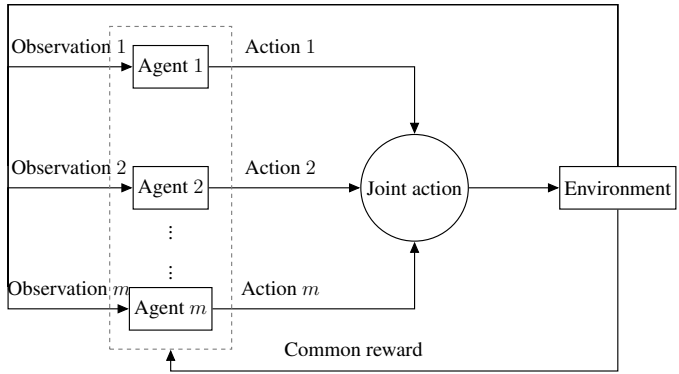
\begin{figure}[ht!]
  \centering
  \captionsetup{justification=centering,margin=2cm}
  \resizebox{.75\textwidth}{!}{%
\begin{tikzpicture}[auto, thick, node distance=2cm, >=triangle 45]
    \draw
        node at (15,-5) [block] (E) {\Large Environment}
        node at (4, -2) [block] (a1) {\Large Agent $1$}
        node at (4, -5) [block] (a2) {\Large Agent $2$}
        node at (4, -6) (dots1) {\Large $\vdots$}
        node at (4, -7) (dots2) {\Large $\vdots$}
        node at (4, -8) [block] (am) {\Large Agent $m$};

    \node[circle,draw,thick,minimum size=1cm] at (10,-5) (ja) {\Large Joint action};
    \coordinate[above = 4cm of E] (v1);
    \coordinate[left = 15cm of v1] (v2);
    \coordinate[below = 4cm of v2] (v3);
    \coordinate[below = 7cm of v2] (v4);
    
    \draw[->] (E) -- (v1) -- (v2) -- (v3) |- (a1);
    \draw[->] (E) -- (v1) -- (v2) -- (v4) |- (a2);
    \draw[->] (E) -- (v1) -- (v2) -- (v4) |- (am);
    \draw[->] (ja) -- (E);
    \draw[->] (a1) -| (ja);
    \draw[->] (a2) -- (ja);
    \draw[->] (am) -| (ja);

    \node[] at (6.7,-1.5) {\Large Action $1$};
    \node[] at (6.7,-4.5) {\Large Action $2$};
    \node[] at (6.7,-7.5) {\Large Action $m$};

    \node[] at (1.5,-1.5) {\Large Observation $1$};
    \node[] at (1.5,-4.5) {\Large Observation $2$};
    \node[] at (1.5,-7.5) {\Large Observation $m$};

    \node (rect) at (4,-5) [draw,dashed,color=gray,thick,minimum width=3cm,minimum height=8cm] {};

    \coordinate[below = 4cm of E] (v5);

    \draw[->] (E) -- (v5) -| (rect);

    \node[] at (8.5,-9) {\Large Common reward};
\end{tikzpicture}
}
\caption{The multi-agent interaction with the vehicle environment.}
\label{fig:mdp}
\end{figure}

Mathematically, at each time-slot $t$, given that the environment is in the state $\mathscr{S}(t)$, each $v$ transmitting vehicle successively (i) receives an observation $\mathscr{O}_v(t)$ from the environment, (ii) takes an action $a_v(t)$, and (iii) receives a reward $r(t+1)$. Then, the environment evolves to the next state $\mathscr{S}(t+1)$. The reward function $r(t+1)$ is independent of agent $v$, meaning that agents receive a common reward signal. This is designed to encourage cooperative behavior among  agents. The reward function $r(t+1)$ is computed as soon as all agents take their actions, forming a common action $\mathbf{a}(t)=(a_v(t),\mathbf{a}_{-v}(t))$, where $\mathbf{a}_{-v}(t)$ denotes the actions of all vehicles but $v$. In the following, we describe the key elements of the multi-agent MDP.

\subsection{The State Space}
The state $\mathscr{S}(t)$ of the vehicle environment at time-slot $t$ is not known directly by each transmitting vehicle (called agent from now on). Instead, given the state $\mathscr{S}(t)$, each agent $v$ receives an observation $\mathscr{O}_v(t)$. State $\mathscr{S}(t)$, which is hard to acquire by each agent due to its high mobility, includes accurate and global information about CSI of all agents and their transmission behaviors (e.g., the remaining data bits to be sent in future time-slots). However, each agent $v$ can gain knowledge only by observing the environment using $\mathscr{O}_v(t)$. In our model, $\mathscr{O}_v(t)$ includes local CSIs and transmission behavior of $v$ alone. Specifically, $\mathscr{O}_v(t)$ is given by the following tuple:
\begin{align}\label{obs}
    \mathscr{O}_v(t)=(\mathbf{g}_{v}(t), \mathbf{d}'_{v}(t), \mathbf{z}(t), \mathbf{l}_v(t), s_v, e_v, k, \epsilon),
\end{align}
where $\mathbf{g}_{v}(t)\coloneqq(d_{v,w}(t): w\in\mathscr{W}_v)$ represents the large-scale and small-scale fading between agent $v$ and all other receiving vehicles $\mathscr{W}_v$. The value $\mathbf{g}_v(t)$ can be accurately estimated by the receiving vehicles and returned to each agent without significant delay~\cite{8792117}. Similarly, $\mathbf{d}'_{v}(t)$ is the distance between agent $v$ and all other agents. The distances $\mathbf{d}'_v(t)$ help agent $v$ observe the location of other agents (who are likely to interfere with it) and thus contribute to the required cooperative behavior. The distance $\mathbf{d}'_v(t)$ is location-based channel information that can be accurately estimated by agent $v'$ and can be fed back to $v$ with a delay-free communication~\cite{8792117}. Therefore, it can be assumed that $\mathbf{g}_v(t)$ and $\mathbf{d}'_v(t)$ are instantly available to agent $v$. In addition, $\mathscr{O}_v(t)$ includes a decision variable $\mathbf{z}(t)$ that indicates whether or not the agents were transmitting in previous time-slots and if so, which packet did they transmit. The fourth observation that agent $v$ acquires is $\mathbf{l}_v(t)$ which indicates the remaining bits of the two packets that agent $v$ should send (e.g., initially, $l_v^{i}(t)=\sigma_v^{i}$). The observation also includes the global start time $s_v$ and end time $e_v$ to primarily capture the safety packet's deadline. The last two parameters that $\mathscr{O}_v(t)$ includes are the learning episode index $k$ and the exploration rate of exploration policy $\epsilon$-greedy. These last two parameters are used mainly to remove the non-stationarity problem that arises in multi-agent DQL~\cite{pmlr-v70-foerster17b}. In fact, DQL relies on an experience replay memory that is used to better train the DQN by creating a data set and applying batch sampling from time to time. As discussed earlier, if DQL is applied to each agent independently, non-stationarity occurs and batch sampling no longer reflects the current dynamics of the environment. In~\cite{pmlr-v70-foerster17b}, the authors solved the non-stationarity problem by conditioning each agent's Q-function on a fingerprint that tracks the policies of the other agents. They showed that the agent's policy is strongly correlated with the training episode index $k$ and the exploration rate $\epsilon$.

\subsection{The Action Space}
The VRA problem is solved online, i.e. on a time-slot basis. At each time-slot, the agent has to make four decisions which are given as follows: (i) coverage broadcast selection, (ii) packet selection, (iii) RB allocation, and (iv) power allocation. The first decision is the selection of the set of receiving vehicles to broadcast to. The second decision is the selection of the packet to broadcast (safety or non-safety). The third decision is the selection of the frequency-slot and time-slot to be used for transmission. Finally, the fourth decision is the allocation of transmission power for transmitting vehicles. For (i), we define a set of coverages $\mathscr{C}\coloneqq\{c_1,c_2,\ldots,c_\delta\}\cup\{0\}$ from which the coverage of each transmitting vehicle could be selected. Thus, if an agent chooses $c_i\mathscr{C}$ for a certain $i\in\{1, 2,\ldots,\delta\}$, then it will broadcast to all the receiving vehicles that are present in the circle of radius $c_i$ (note that when $c_i=0$, the transmitting vehicle will not broadcast any packet). For (ii), we define the set $\mathscr{B}\coloneqq\{\emptyset,\{1\},\{2\},\{1,2\}\}$ that indicates which packet of which slice the agent will transmit. The agent has four possible choices: it does not transmit, it transmits packet 1 (from slice 1), it transmits packet 2 (from slice 2), or it transmits both packets (from slices 1 and 2). For (iii), the RB allocation consists in choosing the frequency-slot $f$ to use in the current time-slot $t$. The last decision is (iv) which represents the power allocation. Normally, power should be allocated from a continuous interval $[0,p_v^{\textsf{max}}]$ which makes DQL implementation more complex. Nonetheless, discrete power allocation is a realistic assumption in many real systems~\cite{8052127,7010886,4215595}. Thus, to keep things simple and realistic and following the assumptions of previous work~\cite{8792382}, we define a set $\mathscr{P}=\{p_1,p_2,\ldots,p_\ell\}$ of discrete power levels that the agent can use to transmit its packets on the allocated RB.

In summary, the action space of agent $v$ at time-slot $t$ is given by the set $\mathscr{A}_v(t)\coloneqq\mathscr{C}\times\mathscr{B}\times\mathscr{F}\times\mathscr{P}$. 

\subsection{The Reward Function}
As discussed in 5G-NR C-V2X, an important optimization objective in vehicular broadcast communication is the PRR, which is directly related to the number of successfully received packets. In VRA, we correlate the objective function in~\eqref{pb:1} with the reward function. The flexibility in the design of reward functions is what makes RL a particularly attractive approach to solving NP-hard problems like VRA.

The reward of agent $v$ at time-slot $t$ depends on whether the agent successfully transmits the packet or not. As long as the packet is not successfully transmitted, the reward is set to the normalized achievable rate between agent $v$ and each of its selected receiving vehicles. (Normalization is used to make the reward less than one.) As soon as the agent successfully transmits its packet, the reward is set to $1$ (the highest possible reward for agent $v$). Setting the reward to a number less than $1$ when the agent has failed to transmit its packet helps the agent to acquire useful informations for its future decisions. Furthermore, by setting the maximum possible reward once the packet has been successfully transmitted, the agent learns the best possible outcome. Agent's $v$ individual reward at time-slot $t$ when transmitting packet $i$ to $w$ using frequency-slot $f$ is given by the following formula:
\begin{align}\label{rew}
    \begin{cases}
    \frac{\beta\tau}{\Gamma}\lg(1+\textsf{sinr}_{v,w,f}^i(t)), & \text{if $v$ does not finish its transmission}, \\
    1, & \text{otherwise},
    \end{cases}
\end{align}
where $\Gamma$ is a normalization constant to make the reward less than $1$.

The reward function $r(t)$ is thus the sum of all individual rewards of each agent; creating a common reward for all agents. This will encourage agents to play cooperatively while learning, in order to maximize their rewards.

The goal of the DQL is to maximize the cumulative (discounted) reward in the long run, given an initial state of the environment $\mathscr{S}(0)$. This cumulative reward is given mathematically by:
\begin{align}
    \sum_{\iota=0}^{\infty}\gamma^kr(t+\iota+1),\quad0\leq\gamma\leq1,
\end{align}
where $\gamma$ is called the discount factor.

\subsection{The Transition Probability Function}
The evolution of the environment from one state to another is often modeled by a probability distribution function. It is given mathematically by the probability of moving to  state $\mathscr{S}(t+1)$ and obtaining reward $r(t+1)$ given that the environment was in the state $\mathscr{S}(t)$ and that the action taken was $a(t)$, i.e. $\Pr[\mathscr{S}(t+1),r(t+1)|\mathscr{S}(t),\mathbf{a}(t)]$. This probability function depends on the highly dynamic vehicular environment (depends on channel coefficients and vehicle motion) and cannot be computed explicitly due to the complex nature of the vehicular environment.

\subsection{The Training Phase of DQL}
As discussed previously, DQL is composed of two phases: the learning phase and the inference phase. The learning phase lasts a number of episodes where each episode spans a time horizon of $T$ time-slots. DQL uses deep neural networks to approximate the Q function. We leverage DQL with prioritized replay memory~\cite{2015arXiv151105952S} and a dueling~\cite{pmlr-v48-wangf16}. In general, experience replay memory helps agents remember and use past experiences. Standard replay memory is used to sample experience transitions uniformly, without paying attention to the importance of the sampled experiences. Note, however, that prioritized experience replay memory is proposed to pay more attention to important experiences. This indeed allows agents to learn better~\cite{2015arXiv151105952S}. On the other hand, dueling~\cite{pmlr-v48-wangf16} is proposed as a new neural network architecture that represents two estimators for the Q-function.

The proposed DQL algorithm uses a DQN with a weight vector $\mathbf{w}_v$ to represent the Q-function of each agent $v$. In other words, we create $m$ DQNs; one for each agent. The input of DQN $v$ is given by the observation $\mathscr{O}_v(t)$ obtained by observing the state of the vehicle environment $\mathscr{S}(t)$. The output of the DQN $v$ is the value of the Q function which is given by the appropriate action taken by agent $v$ among the set of possible actions $\mathscr{A}_v(t)$. All DQNs are trained simultaneously. Note that we assume some level of synchronization between agents, which can be achieved by a central controller (or a roadside unit) that collects agents' actions and forms the joint action to calculate the common reward function. In fact, once each agent chooses its action, it sends it to the roadside unit (RSU). Next, the RSU forms a joint action composed of agents’ individual actions. It can then check if the joint action is feasible or not according to the problem constraints and can calculate the common reward function.

The learning phase lasts $H$ episodes. In each episode, agents explore the action space using the $\epsilon$-greedy policy. Each episode $k$ covers a time horizon of $T$ time-slots. At the beginning of the first time-slot, the initial state of the vehicular environment (initial distances of vehicles, etc.) is revealed to all agents. For each time-slot $t\geq1$, each agent chooses an action, that is a tuple $a_v(t)\coloneqq(c_v,b_v,f_v,p_v)\in\mathscr{A}_v(t)$, according to $\epsilon$-greedy, where $c_v$ is the coverage area, $b_v$ is the packet(s) to be sent, $f_v$ is the frequency-slot and $p_v$ is the transmission power. Once all agents have chosen their actions, a joint action $\mathbf{a}(t)$ is formed and the reward function $r(t+1)$ is calculated at the next time-slot $t+1$. Each agent moves to the next state due to the evolution of the channel coefficients and vehicle mobility. The resulting tuple $\text{Exp}_v\coloneqq(\mathscr{O}_v(t),a_v(t),r(t+1),\mathscr{O}_v(t+1))$ is called the experience of agent $v$ and is stored in its prioritized replay memory with some associated priority. After a few episodes, a mini-batch $\mathscr{M}_v$ is sampled according to the priorities from the prioritized replay memory. This mini-batch is used to update the DQN weight parameter using a variant of the stochastic gradient descent algorithm to minimize the loss function. The loss function is given by the mean square error as follows: 
\begin{align}\label{loss}
    \sum_{\text{Exp}_v\in\mathscr{M}_v}\Bigl[r(t+1)+\gamma\max_{a_v}\{Q(\mathscr{O}_v(t+1),a_v;\mathbf{w}^{-}_v)\}-Q(\mathscr{O}_v(t),a_v(t);\mathbf{w}_v)\Bigr]^2,
\end{align}
where each DQN is represented mathematically by the Q function $Q(\mathscr{O}_v(t),a_v(t);\mathbf{w}_v)$ (that DQL tries to approximate) and $\mathbf{w}^{-}_v$ is the weight parameter of a duplicate copy of the original DQN (the target DQN) that is created in order to update the original DQN from time to time. The creation of a target DQN is suggested by the quasi-static target network method~\cite{mnih-dqn-2015} to set the targets of the Q values. The pseudocode for the learning phase of the DQL algorithm is given in Algorithm~\ref{alg:dql}.
\begin{algorithm}[ht!]
  \caption{The Training Phase of DQL}
  \label{alg:dql}
  \begin{algorithmic}[1]
    \Require{Agents and environment.}
    \Ensure{Trained DQNs.}
    \State Start simulator: generate vehicles and network parameters.
    \State Initialize the DQN of each agent $v$.
    \For{each episode $k\gets1$ \textbf{to} $H$}
        \State Reset and build the environment.\Comment{Generate vehicles, their speeds and directions, generate network parameters and QoS requirements, etc.}
        \State Anneal the exploration rate $\epsilon$ for each agent.\Comment{Use the decayed $\epsilon$-greedy approach in~\eqref{decayedeps}.}
        \For{each time-slot $t\gets1$ \textbf{to} $T$}
            \State $t_k\gets(k-1)T+t$
            \For{each agent $v$}
                \State Observe $\mathscr{O}_v(t)$.\Comment{The observation includes all vehicular network parameters according to~\eqref{obs}.}
                \State Choose an action $a_v(t)$ using $\epsilon$-greedy.
            \EndFor 
            \State Obtain the joint action $\mathbf{a}(t)$.\Comment{The action of all agents.}
            \State Find the common reward $r(t+1)$.\Comment{The reward is designed as in~\eqref{rew}.}
            \For{each agent $v$}
                \State Obtain the next observation $\mathscr{O}_v(t+1)$.
                \State Prioritize the experience $\text{Exp}_v$.\Comment{For any agent, its $t$th experience is assigned a priority $\pi_t$ as given in~\eqref{priority1}.}
                \State Store $\text{Exp}_v$ in $\mathscr{M}_v$.
                \If{$t_k \text{ mod } T_{\text{train}}=0$}
                    \State Sample a mini-batch from $\mathscr{M}_v$.\Comment{For any agent, its $t$th experience is sampled according to its probability $\Pi_t$ given in~\eqref{priority1}. We sample according to these probabilities $N_{\text{samples}}$.}
                    \State Do a mini-batch training.\Comment{After constructing the dataset of samples, each agent starts the training by applying the stochastic gradient descent algorithm to minimize its loss function according to~\eqref{loss}.}
                    \State Update the priorities.\Comment{The priority of every experience $t$ is updated according to the temporal difference $\Delta_t$ as in~\eqref{td}.}
                \EndIf
                \If{$t_k \text{ mod } T_{\text{target}}=0$}
                    \State Update the target DQN of $v$.
                \EndIf
            \EndFor
        \EndFor
    \EndFor
  \end{algorithmic}
\end{algorithm}

The training phase of DQL requires as input the vehicular environment which includes vehicles, channel coefficients, packet requirements, available RBs and any other relevant network parameters. It returns as output trained DQN of each agent. First, DQL generates all the vehicular network parameters, and then it initializes the weights of each DQN. Then, it iterates the episodes. For each episode $k$, the vehicular environment is constructed by (i) updating the network parameters, e.g., the remaining bits of each agent are updated based on the previous episodes, and (ii) moving the vehicles according to the mobility model. Then, the exploration rate $\epsilon$ is annealed based on the episode index. The annealing of the exploration rate over time is a technique used in RL to solve the dilemma between exploration and exploitation, i.e., over time we decrease $\epsilon$ to increase the probability of exploitation when the agent starts to learn something useful. To perform annealing, we use the decayed $\epsilon$-greedy algorithm which works as follows. Let $\epsilon_{\text{start}}$ be the initial value of $\epsilon$ (before the first episode) and let $\epsilon_{\text{end}}$ be the final value of $\epsilon$. Let $H_{k}$ be the number of learning episodes after which the annealing stops. Thus, the value of $\epsilon$ is updated (annealed) according to:
\begin{align}\label{decayedeps}
	\begin{cases}
		\epsilon\gets(\epsilon_{\text{start}}-\epsilon_{\text{end}})(1-k/H_k)+\epsilon_{\text{end}},&\text{ if }  k<H_k\\
		\epsilon\gets\epsilon_{\text{end}},& \text{ otherwise.}
	\end{cases}
\end{align}
After $H_k$ episodes, the value of $\epsilon$ is no longer decreased and is set to $\epsilon_{\text{end}}$. 

Once all agents have chosen their actions based on the annealed $\epsilon$ as in~\eqref{decayedeps}, the joint reward is computed by each receiving vehicle through the physical sidelink feedback channel (PSFCH)~\cite{3gpp.38.885}. Specifically, a receiving vehicle calculates the received \textsf{sinr} (and thus the achievable data rate) and finds out how many bits a particular agent is transmitting. Then it broadcasts this information. In this way, all agents can acquire the joint action $\mathbf{a}(t)$ at any time $t$. To guarantee some level of synchronization, the joint action can be obtained by an RSU that collects all individual actions. Once the reward signal is received by all agents, learning begins. Then, the environment moves to the next state according to vehicle mobility and channel variations. Each agent then adds its experience $\text{Exp}_v(t)$ to its prioritized replay memory. Initially, agents assign random priorities to their experiences, but the priorities change as agents begin to learn and update the parameters of their DQNs. Specifically, each agent's $t$ experience is assigned a priority $\pi_t$ (the agent index is omitted for simplicity). Then, experience $t$ is sampled from each agent's replay buffer with a probability $\Pi_t$ given by~\cite{2015arXiv151105952S}:
\begin{align}\label{priority1}
	\Pi_t\coloneqq\dfrac{\pi_t^\alpha}{\sum_{t'}\pi_{t'}^\alpha}.
\end{align}
We use the proportional prioritization approach where $\pi_t=|\Delta_t|+\mu$, with $\Delta_t$ represents the time difference error of experience $t$ and $\mu$ is a small positive constant used to avoid the limiting case of transitions that are not revisited once their error is zero. Finally, importance sampling is used to eliminate any possible bias introduced by prioritization. This is achieved by introducing a weight called the importance sampling weight of experience $t$ given by $w_t\coloneqq(E\cdot\Pi_t)^{-\theta}$ with $E$ being the number of experiences in the replay buffer.

Once every $T_{\text{train}}$, each agent samples a mini-batch of size $N_{\text{samples}}$ from its prioritized replay memory. This prioritized replay memory forms a dataset that each agent uses to perform learning. Indeed, each agent uses a well-known variant of stochastic gradient descent to minimize the mean square error (or loss) which is given by~\eqref{loss}. In~\eqref{loss}, the term $r(t+1)+\gamma\max_{a_v}\{Q(\mathscr{O}_v(t+1),a_v;\mathbf{w}^{-}_v)\}$ denotes the target value that each DQN $v$ tries to reach or adjust. Then, each agent updates the priorities of the sampled experiences in proportion to the value of the loss. In other words, the priority of experience $t$ is updated as $\pi_t=|\Delta_t|+\mu$ where $\Delta_t$ is the time difference error and is defined as follows:
\begin{align}\label{td}
	\Delta_t=r(t)+\gamma\max_{a}\{Q(\mathscr{O}(t),a;\mathbf{w}^{-})\}-Q(\mathscr{O}(t-1),a(t-1);\mathbf{w}),
\end{align}
where agent's index is omitted. Finally, once every $T_{\text{target}}$ each trained DQN is copied into the target DQN.

\subsection{The Inference Phase of DQL}
The inference phase of DQL is as follows. First, the trained DQNs (their weight parameters) are loaded. Similarly, the annealed $\epsilon$ is loaded from the last training episode (the episode index is also revealed). Then, for each episode, which now represents a random realization of the channel, the environment is reset and built---initializing the network parameters and transmission behaviors of each agent. Then, for each time-slot, each agent $v$, after observing the environment, chooses the best action which is given by the maximum value of its Q function approximated by its DQN. Once all agents have chosen their actions, a common action is formed by the feedback from the receiving vehicles to the transmitting vehicles. Then the common reward is obtained and the next episode begins.

The inference phase of DQL is an online distributed algorithm that is run in each time-slot to find the best possible action to select without knowing the future observation. The learning phase is the most computationally intensive task in DQL. It is executed for a large number of episodes and can be performed offline with different channel conditions and network topologies. The pseudo-code for the inference phase of DQL is given in Algorithm~\ref{alg:dql2}.
\begin{algorithm}[ht!]
  \caption{The Inference Phase of DQL}
  \label{alg:dql2}
  \begin{algorithmic}[1]
    \Require{The trained DQNs.}
    \Ensure{Resource allocation solution of VRA.}
    \State Load the DQNs---one DQN per agent.\Comment{After training and saving DQNs' weights, they are loaded to use and apply to learned policy.}
    \For{each episode $k\gets1$ \textbf{to} $H$}
        \State Reset and build the environment.\Comment{Similar to the training phase, we generate vehicles, their speeds and directions, and we generate network parameters and QoS requirements, etc.}
        \State Get the annealed $\epsilon$ for each agent.\Comment{Obtain the saved decayed $\epsilon$ value.}
        \For{each time-slot $t\gets1$ \textbf{to} $T$}
            \For{each agent $v$}
                \State Observe $\mathscr{O}_v(t)$.\Comment{The observation includes all vehicular network parameters according to~\eqref{obs}.}
                \State Choose $a_v(t)$ that maximizes the Q function.\Comment{Choose the action $a_v=\argmax_a{Q(\mathscr{O}_v,a;\mathbf{w}_v)}$.}
            \EndFor 
            \State Obtain the joint action $\mathbf{a}(t)$.
            \State Find a solution to VRA.\Comment{Based on the chosen action, calculate the objective function in~\eqref{pb:1a}.}
        \EndFor
    \EndFor
  \end{algorithmic}
\end{algorithm}

Note that the complexity of DQL is dominated by training. This is due to the computationally intensive tasks that must be performed offline to learn the optimal policy. Nevertheless, once learning is done offline, DQL can be used online to act suboptimally according to the learned policy. The complexity of DQL depends on many parameters, namely the learning time of DQNs, the computation time to obtain the reward of each agent, the buffer sampling time, etc. To quantify the worst-case time complexity, we proceed as follows. The most computationally expensive instructions in DQL are written in lines 13, 19 and 20 of Algorithm 1. Let $T_{13}$, $T_{19}$, and $T_{20}$ be the worst-case computation time of these lines, respectively. The worst-case time $T_{13}$ corresponds to the worst-case time to compute the reward, which mainly depends on the SINR computation. Thus, at each time-slot $t$, $T_{13}=\mathcal{O}(m^2nF)$ according to the SINR equation in (1). The worst-case time $T_{19}$ is related to the sampling of a set of experiences in the mini-batch. It is mainly proportional to the size of the sampled mini-batch. Thus, $T_{19}=\mathcal{O}(N_{\text{samples}})$, where $N_{\text{samples}}$ represents the number of sampled experiences. Finally, the time complexity of line 20 in the worst case, $T_{20}$, depends mainly on the training performed. We have a mini-batch of size $N_{\text{samples}}$ that is trained using the ADAM optimizer. Thus, $T_{20}$ depends mainly on the ADAM optimizer, the training rate used, and the number of samples $N_{\text{samples}}$. Nevertheless, the value of $T_{20}$ cannot be given explicitly. Overall, DQL has a complexity of $\mathcal{O}(HT(T_{13}+mN_{\text{samples}}T_{20}))=\mathcal{O}(HT(m^2nF+mN_{\text{samples}}T_{20}))$, where $H$ and $T$ are respectively the number of episodes and the number of time-slots.

\section{Simulation Results}\label{sec:sim}
In this section, we validate the proposed resource allocation method. We present simulation results to illustrate the performance of the proposed DQL algorithm. The simulation setup is based on the road configuration for the highway scenario detailed in 3GPP TR 37.885~\cite{3gpp.37.885} and also used, to name a few, in~\cite{7913583}. We consider a multilane highway with a total length of $2$ km where each lane is $4$ m wide. There are a total of six lanes---three for the forward direction (vehicles move from right to left) and three for the reverse direction (vehicles move from left to right). Transmitting and receiving vehicles are generated in the vehicular environment by a spatial Poisson process. The vehicle speed determines the vehicle density and we consider an average inter-vehicle distance (in the same lane) of $2.5\times V$~\cite{7974737} where $V$ is the absolute vehicle speed. The speed of a vehicle depends on the lane it is in: the $i$th forward lane (top to bottom with $i\in\{1,2,3\}$) is characterized by the speed of $60+2(i-1)\times10$ km/h, while the $i$th backward lane (top to bottom with $i\in\{1,2,3\}$) is characterized by the speed of $100-2(i-1)\times10$ km/h. The number of transmitting vehicles $m$ and receiving vehicles $n$ is chosen randomly from the generated vehicles. The important simulation parameters are given in table~\ref{my-label}. In Fig.~\ref{fig:setup1}, we take a snapshot of the initial vehicle topology at the beginning of the first time-slot where the vehicles are randomly chosen according to the spatial Poisson process.\if and Fig.~\ref{fig:setup2}.\fi
\begin{table}[htp!]
\centering
 \caption{Vehicular network parameters}
\label{my-label}
\resizebox{0.75\columnwidth}{!}{%
\begin{tabular}{||l||l||}
\hline
Parameter                         & Value \\ [0.5ex]
\hline\hline       
Carrier frequency                 & $2$ GHz \\ \hline
Bandwidth per RB                  & $1$ MHz \\ \hline
Vehicle antenna height            & $1.5$ m \\ \hline
Vehicle antenna gain              & $3$ dBi \\ \hline
Vehicle receiver noise figure     & $9$ dB \\ \hline
% \addlinespace
Shadowing distribution            & Log-normal \\\hline
Fast fading                       & Rayleigh fading \\\hline
Pathloss model                    & LOS in WINNER + B1~\cite{winner} \\\hline
Shadowing standard deviation      & $3$ dB \\ \hline
Road configuration                & Highway road configuration~\cite{3gpp.37.885} \\ \hline
Vehicle drop model                & Spatial Poisson process \\\hline
Noise power $N_0$                 & $-114$ dBm \\ [1ex]
\hline
\end{tabular}
}
\end{table}

\begin{figure}[hb!]
        \centering
        \includegraphics[width=0.75\textwidth]{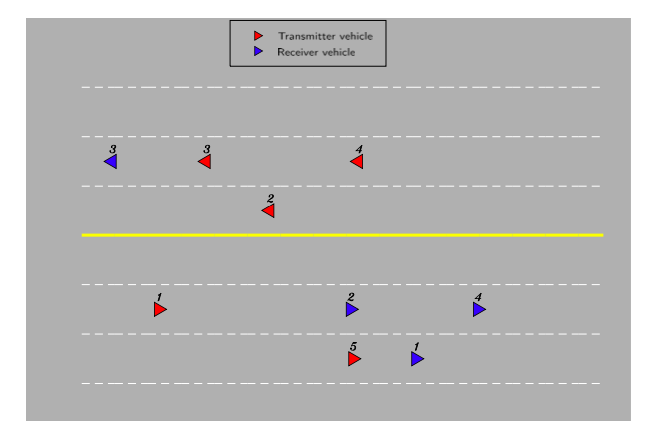} % first figure itself
        \caption{A snapshot of the vehicle topology at the beginning of the first time-slot with five transmitter vehicles and four receiver vehicles.}
        \label{fig:setup1}
\end{figure}

Unless otherwise specified, the packet requirements of the non-safety slice are generated uniformly at random in $\{0.1..1\}$ Mbit. The packet requirements of the safety slice are equal to $1200$ bytes. Each agent can choose a coverage (in m) from the set $\mathscr{C}=\{100,200,400,800,1000,1200,1400\}\cup\{0\}$. The power levels (in dBm) are given by the set $\mathscr{P}=\{5,10,15,20,23,27,30\}\cup\{-100\}$ where $-100$ dBm is used to indicate that the corresponding agent will not transmit its packets (similar to the zero coverage).

We train three different vehicular networks for a total time horizon of $100$ ms. The first one is called (2,2,1,5)-network and consists of $m=2$ agents (transmitting vehicles), $n=2$ receiving vehicles, $F=1$ frequency-slots and $T=5$ time-slots, each with a duration of $20$ ms. The second is called (5,4,2,10)-network and consists of $m=5$ agents, $n=4$ receiving vehicles, $F=2$ frequency-slots and $T=10$ time-slots, each with a duration of $10$ ms. The third is called (6,4,4,20)-network and consists of $m=6$ agents, $n=4$ receiving vehicles, $F=4$ frequency-slots and $T=20$ time-slots, each with a duration of $5$ ms. We use the first small network for optimal comparison purposes. We implement the mixed integer nonlinear programming (MINLP) formulation given in~\eqref{pb:1} in the Julia programming language~\cite{bezanson2017julia} using Juniper~\cite{juniper} package. Juniper is a MINLP solver written in Julia that solves MINLPs using NLP solvers and then branch-and-bound (BnB) based solvers (MIP solvers). In our implementation, we used Ipopt as the NLP solver and CPLEX as the MIP solver. DQNs are created and trained in Julia using Flux.jl~\cite{innes:2018} machine learning library. Each DQN consists of an input layer and an output layer and three fully connected hidden layers containing $500$, $350$, and $260$ neurons respectively. The rectified linear unit activation function (ReLU) given by $f(x)=\max\{0,x\}$ is used in each layer. Each DQN is trained with the ADAM optimizer~\cite{adam} with a learning rate of $10^{-5}$. The training lasts $H=3000$ episodes with an exploration rate starting at $\epsilon_{\text{start}}=1$ and annealed to $\epsilon_{\text{end}}=0.02$ for the $80\%$ of the episodes (i.e., from episode $k>2400$, the exploration rate is fixed at $0.02$). The target update frequency and training frequency ($T_{\text{target}}$ and $T_{\text{train}}$) are set according to the trained network, so that every $10$ episodes, DQL performs a mini-batch training and every $100$ episodes, it performs a target DQN update. The size of the mini-batch is chosen to be large and equal to $2000$~\cite{l.2018dont} and the experiment priority update parameters are set to $\alpha=0.6$, $\theta=0.4$, and $\mu= 0.001$~\cite{2015arXiv151105952S}. The channel coefficients (including slow and fast fading) change over time at each stage of the learning phase. The requirements and deadlines of the safety packets as well as the requirements of the non-safety packets are fixed in the training phase and change in the inference phase to validate the robustness of our proposed approach.

As discussed earlier, to the best of our knowledge, there is no current research work that solves VRA while considering network slicing selection, coverage selection, RB and power allocation. Due to the lack of comparisons, we adopt the following benchmarks to compare our proposed algorithm. We implement three benchmarks: two are based on NOMA technique and one is based on OMA technique. The partial idea of all benchmarks comes from~\cite{7974737,8632657} which is based on transforming the RB allocation problem in VRA into a matching problem. Then, a matching (or rotation) algorithm is adopted. All benchmarks are centralized and use randomization and offline decisions contrary to the distributed and online nature of DQL. They are called OMA-MP, NOMA-MP, and NOMA-RP. In the case of OMA-MP, each RB is used by no more than one vehicle to meet OMA constraints and the vehicles transmit with their maximum transmission power. In the NOMA-MP and NOMA-RP, each RB can be used by any vehicle and the vehicles transmit with their maximum transmission power and with random transmission power, respectively. Coverage and slot selection are performed randomly at the beginning of the time horizon. The allocation of RBs to vehicles is done in a similar way in all benchmarks~\cite{7974737,8632657}. First, an initial RB allocation is performed that finds the RB that yields the highest sum of channel power gain between a transmitting vehicle $v$ and its corresponding target receiving vehicle in $\mathscr{W}_v$. The sum criterion on receiver vehicles is adapted to our V2V broadcasting case and is used to maximize the number of receiving vehicles that have high channel power gains. Once the initial allocation is obtained, a swap match is performed to improve the number of successfully received packets. If no swap improves the matching, the algorithm terminates, otherwise the algorithm continues swapping operations (see~\cite{8632657} for more details).

\begin{figure}[h!]
  \centering
  \includegraphics[width=0.8\textwidth]{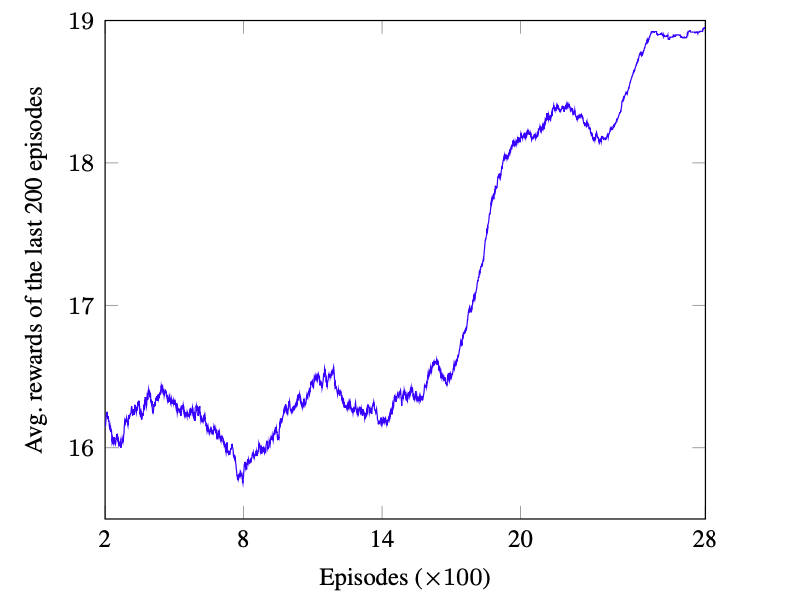}
  \caption{Training rewards.}
  \label{fig:1a}
\end{figure}
\begin{figure}[h!]
  \centering
  \includegraphics[width=0.8\textwidth]{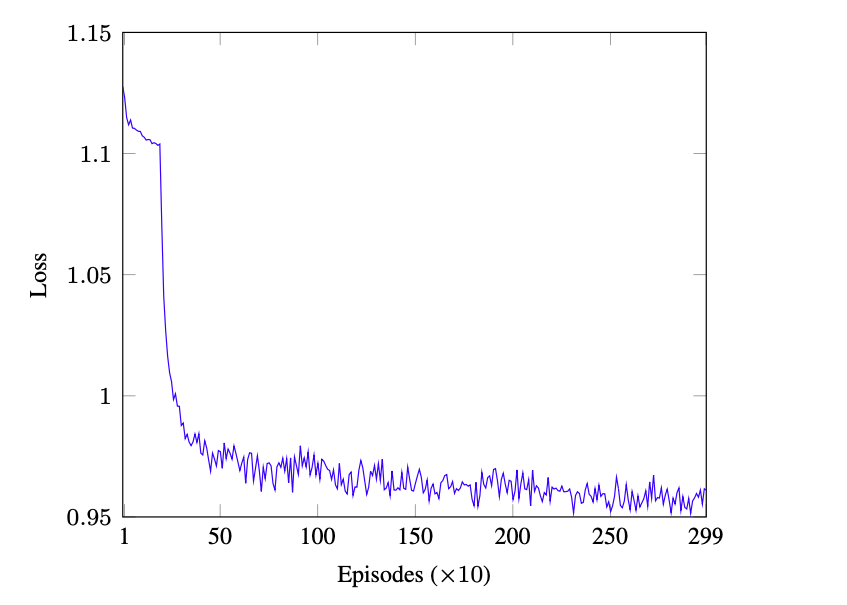}
  \caption{Training loss}
  \label{fig:1aa}
\end{figure}
Fig.~\ref{fig:1a} shows the convergence of DQL as a function of learning episodes. It shows the cumulative average rewards per episode where the average is taken over the last $200$ episodes. It is clear that the average rewards improve as the training episodes increase. This shows the effectiveness of the proposed algorithm. We can see that DQL converges gradually, starting at episode number $2700$. Note that the convergence of the algorithm is not smooth and contains some fluctuations which are mainly due to the high mobility of the vehicular environment. Based on Fig.~\ref{fig:1a}, we concluded that the DQNs should be trained offline for $3000$ episodes, as discussed earlier, to provide some convergence guarantees.

In Fig.~\ref{fig:1aa}, we show the convergence of the loss function used to train each DQN network as shown in~\eqref{loss}. It is clear that the loss decreases to a small value, indicating that the training improves with episode iteration. This makes the proposed multi-agent DQL algorithm important from a practical point of view as it provides some convergence guarantees.

\begin{figure}[h!]
  \centering
  \captionsetup{justification=centering,margin=2cm}
  \includegraphics[width=0.8\textwidth]{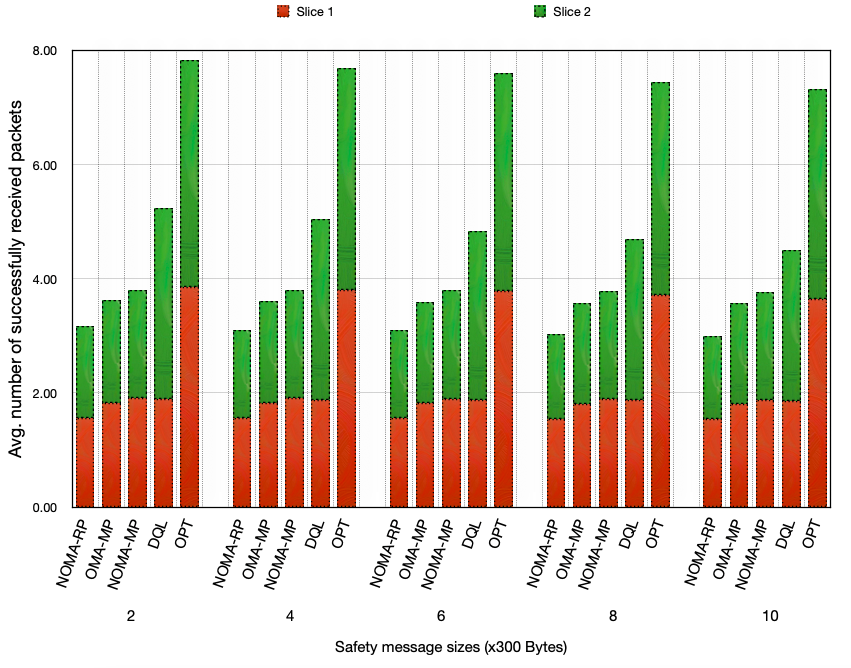}
  \caption{Impact of safety message sizes for the (2,2,1,5)-network.}
  \label{fig:smallnet}
\end{figure}
Fig.~\ref{fig:smallnet} shows the performance of DQL against the optimal algorithm OPT. OPT is obtained with the Julia-based Juniper package using the Ipopt NLP and CPLEX MIP solvers. It is clear that OPT outperforms all the algorithms since it is able to obtain the global optimal solution. However, due to the high complexity of OPT, its implementation in real vehicular network scenarios would be very impractical and that is why Fig.~\ref{fig:smallnet} is only obtained for the small network instance of (2,2,1,5)-network. Fortunately, DQL achieves the best performance compared to other heuristics while being online and distributed. We can also see that DQL is more than $60\%$ close to OPT, which illustrates its outstanding performance.

\begin{figure}[h!]
  \centering
  \captionsetup{justification=centering,margin=2cm}
  \includegraphics[width=0.8\textwidth]{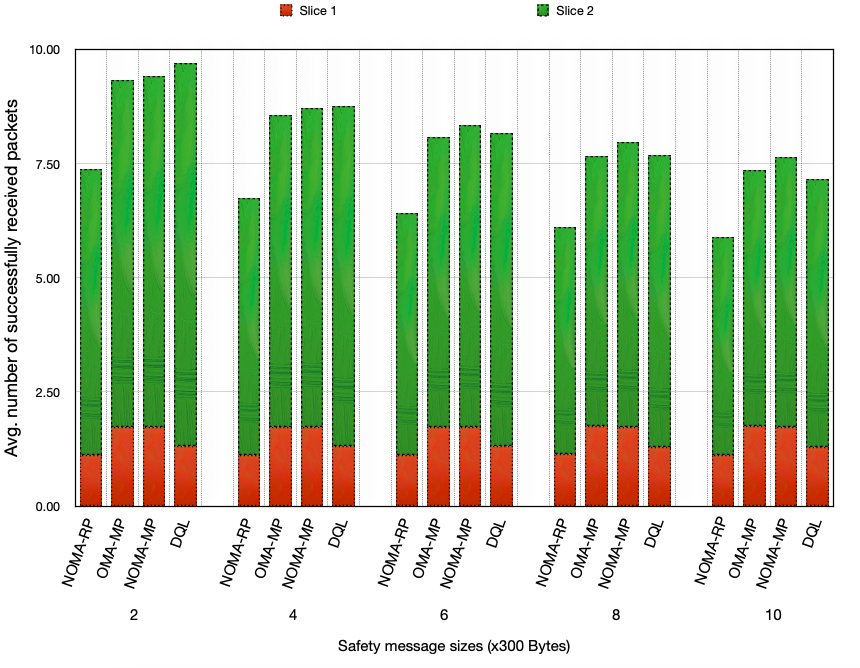}  
  \caption{Impact of safety message sizes for the (5,4,2,10)-network.}
  \label{fig:1b}
\end{figure}
\begin{figure}[h!]
  \centering
  \captionsetup{justification=centering,margin=2cm}
  \includegraphics[width=0.8\textwidth]{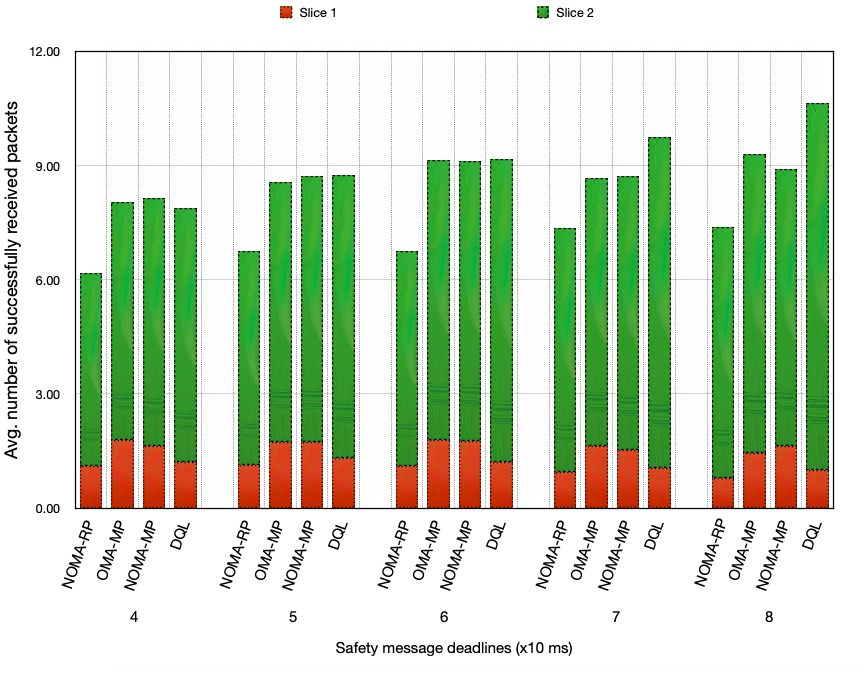}
  \caption{Impact of safety message deadlines for the (5,4,2,10)-network.}
  \label{fig:1c}
\end{figure}
\begin{figure}[h!]
  \centering
  \captionsetup{justification=centering,margin=2cm}
  \includegraphics[width=0.8\textwidth]{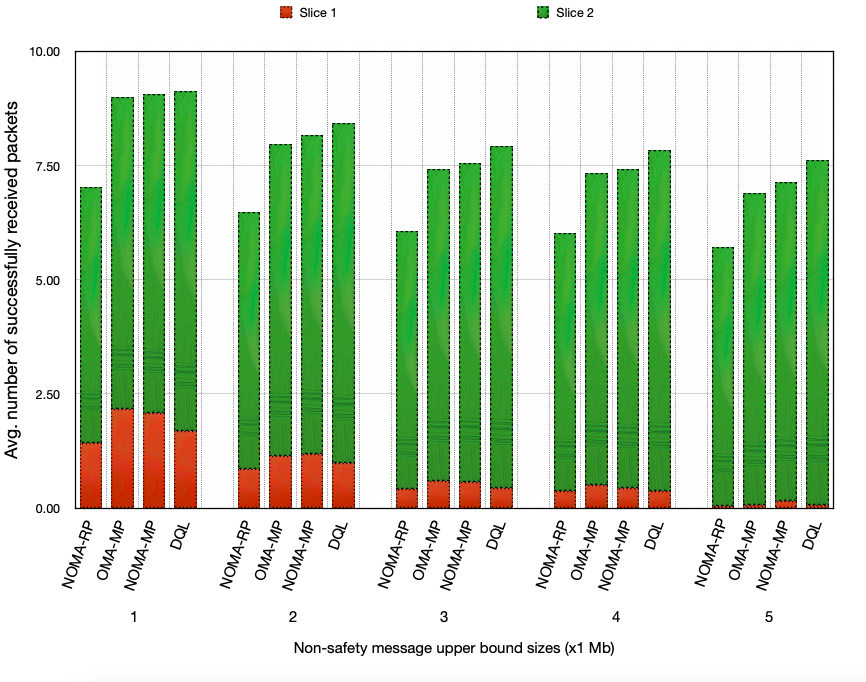}
  \caption{Impact of non-safety message sizes for the (5,4,2,10)-network.}
  \label{fig:1d}
\end{figure}
\begin{figure}[h!]
  \centering
  \captionsetup{justification=centering,margin=1cm}
  \includegraphics[width=0.8\textwidth]{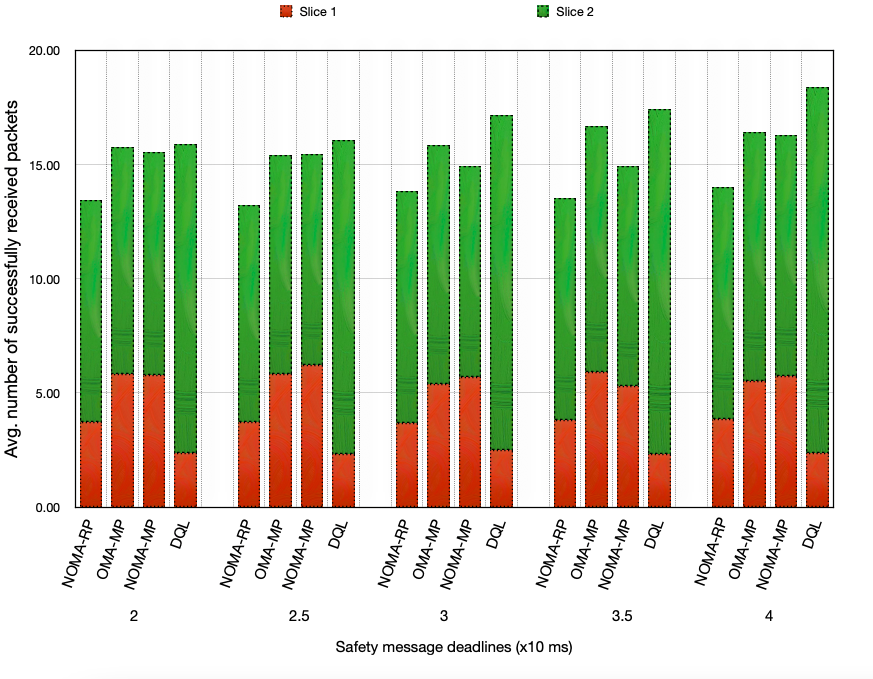}  
  \caption{Impact of safety message deadlines for the (6,4,4,20)-network.}
  \label{fig:2a}
\end{figure}

Fig.~\ref{fig:1b} shows the performance of DQL compared to the benchmarks by varying the packet size of slice 2. We can see that DQL manages to deliver more packets despite the fact that it is distributed and does not have the full and future CSI as in the benchmarks. We can see that NOMA-MP also achieves good performance and has the highest performance among all benchmarks. NOMA-RP achieves the lowest performance, as expected. In addition, DQL achieves a higher number of successfully delivered safety packets compared to non-safety packets. This is especially important in V2V communication because safety packets must have a higher priority to be delivered. We conclude that DQL achieves good performance compared to the benchmarks despite being distributed and runs online without seeing the future.

Fig.~\ref{fig:1c} shows the performance of DQL compared to the benchmarks when varying the deadlines of slice 2 packets. DQL always achieves the best performance when the safety packet deadlines increase. The gap between DQL and the other benchmarks further improves as the deadlines increase. We also notice that NOMA-RP has the worst performance among all algorithms, indicating the need for an appropriate power allocation method in the considered NOMA resource allocation problem.

In Fig.~\ref{fig:1d}, we present the effect of the size of non-safety messages on the performance of DQL. We can see that even if the size of the non-safety packets changes, DQL still performs better compared to other algorithms. This illustrates the effectiveness of the proposed multi-agent method and shows its robustness. Of course, if the size of the non-safety packets increases further, the number of successfully delivered packets will decrease since the safety packets must also be delivered. We conclude that there exists some form of fairness between the safety and non-safety slices that should be carefully studied, which we will do in our future work.

In Fig.~\ref{fig:2a}, we show the performance of DQL in the case of a larger vehicle network, i.e., the (6,4,4,20)-network, with stricter latency requirements (more agents, more RBs and shorter time-slot). Indeed, this figure shows that even after significant system modifications, the proposed DQL multi-agent method is robust and able to guarantee better performance compared to unrealistic centralized benchmarks. We note that a larger vehicle network allows the non-safety slice to deliver more non-safety packets because more resources are available. Thus, DQL better optimizes the utilization of individual RBs to deliver more packets.

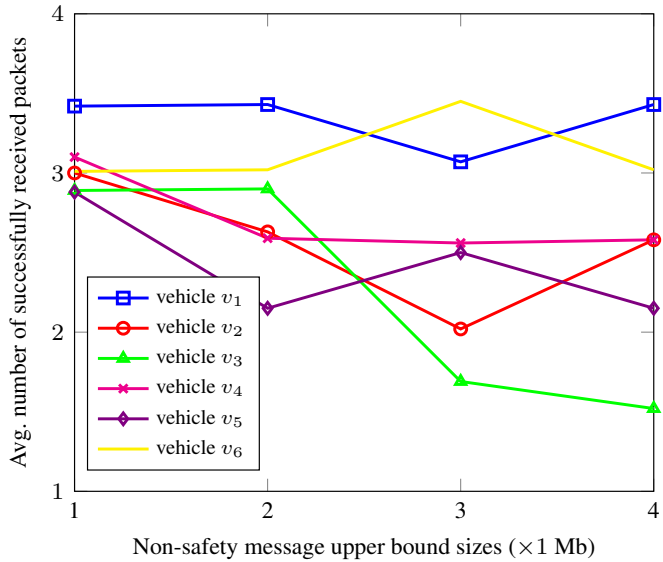
\begin{figure}[h!]
  \centering
  \captionsetup{justification=centering,margin=2cm}
  \resizebox{.75\textwidth}{!}{%
    \begin{tikzpicture}[
        every axis/.style={
        xlabel={Non-safety message upper bound sizes ($\times1$ Mb)},
        ylabel={Avg. number of successfully received packets},
        ymin=1,ymax=4,
        xticklabels={1,2,3,4},
        xtick={1,2,3,4},
        xmin=1,xmax=4,
        x label style={font=\footnotesize},
        y label style={font=\footnotesize}, 
        ticklabel style={font=\footnotesize},
    }]
    \begin{axis}[legend style={at={(0.32,0.45)},font=\scriptsize}]
    \addplot[blue,mark=square,line width=1pt] coordinates {
        (1, 3.4199998) (2, 3.4299998) (3, 3.0700002) (4, 3.4299998)
    };
    \addplot[red,mark=o,line width=1pt] coordinates {
        (1, 3.0) (2, 2.6299999) (3, 2.02) (4, 2.58)
    };
    \addplot[green,mark=triangle,line width=1pt] coordinates {
        (1, 2.8899999) (2, 2.9)  (3, 1.6899999) (4, 1.52)
    };
    \addplot[magenta,mark=x,line width=1pt] coordinates {
        (1, 3.1) (2, 2.59) (3, 2.56) (4, 2.58)
    };
    \addplot[violet,mark=diamond,line width=1pt] coordinates {
        (1, 2.88) (2, 2.1499999) (3, 2.5) (4, 2.1499999)
    };
    \addplot[yellow,line width=1pt] coordinates {
        (1, 3.01) (2, 3.02) (3, 3.4499998) (4, 3.02)
    };
    \legend{vehicle $v_1$, vehicle $v_2$, vehicle $v_3$, vehicle $v_4$, vehicle $v_5$, vehicle $v_6$}
    \end{axis}
    \end{tikzpicture}
  }
  \caption{Individual number of packets for the (6,4,4,20)-network..}
  \label{fig:2b}
\end{figure}
Fig.~\ref{fig:2b} shows the average number of packets delivered by each agent when the size of the non-safety packets varies. Since the non-safety packet requirements are low, each agent delivers almost the same number of packets (except for vehicle $v_1$, all other vehicles deliver almost $3$ packets). This is beneficial since all agents will have almost the same load and thus DQL provides a load balanced solution. As the size of the non-safety packets increases, the requirements are more stringent, and thus to maximize the total number of packets successfully delivered, each DQL agent will deliver a different number of packets, thus unbalancing their loads.

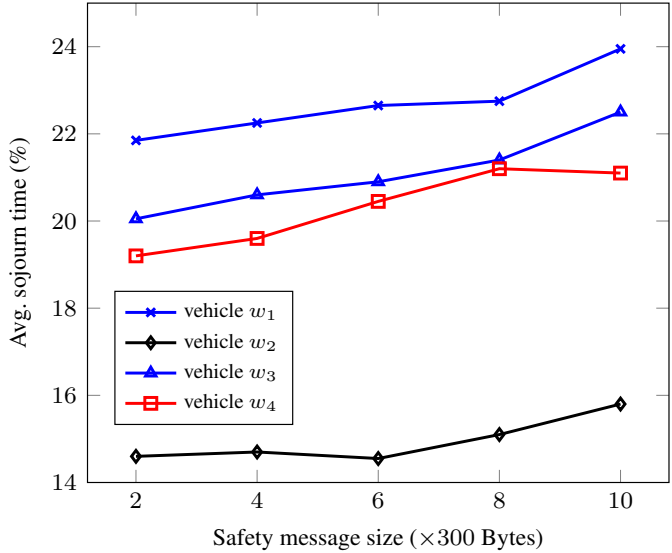
\begin{figure}[h!]
  \centering
  \captionsetup{justification=centering,margin=2cm}
  \resizebox{.75\textwidth}{!}{%
   \begin{tikzpicture}[
    every axis/.style={
        xlabel={Safety message size ($\times300$ Bytes)},
        ylabel={Avg. sojourn time (\%)},
        ymin=14,
        ymax=25,
        xtick={2,4,6,8,10},
        x label style={font=\footnotesize},
        y label style={font=\footnotesize}, 
        ticklabel style={font=\footnotesize},
    }]
    \begin{axis}[legend style={at={(0.2,0.4)},font=\scriptsize,anchor=north,},]
    \addplot[blue,mark=x,line width=1pt] coordinates {
         (2, 21.85) (4, 22.25) (6, 22.65) (8, 22.75) (10, 23.95)    
    };\legend{vehicle $w_1$,vehicle $w_2$,vehicle $w_3$,vehicle $w_4$}
    \addplot[black,mark=diamond,line width=1pt] coordinates {
        (2, 14.60) (4, 14.70) (6, 14.55) (8, 15.10) (10, 15.80)
    };
    \addplot[blue,mark=triangle,line width=1pt] coordinates {
         (2, 20.05) (4, 20.60) (6, 20.90) (8, 21.40) (10, 22.5)
    };
    \addplot[red, mark=square,line width=1pt] coordinates {
        (2, 19.2) (4, 19.6) (6, 20.45) (8, 21.2) (10, 21.1)
    };
    \end{axis}
    \end{tikzpicture}
  }
  \caption{Impact of safety message sizes on sojourn time for the (6,4,4,20)-network.}
  \label{fig:2c}
\end{figure}
In Fig.~\ref{fig:2c}, we present the average DQL sojourn time obtained by each receiving vehicle under the large vehicle network scenario (i.e., the (6,4,4,20)-network). Sojourn time is a metric used in mobile networks to measure ``the expected duration of stay of a mobile node in a particular serving cell''. In our system model, we compute the sojourn time with respect to each receiving vehicle $w$ as the expected time that $w$ remains in the coverage of a particular transmitting vehicle. In other words, once an agent chooses a coverage area in which it communicates with $w$, the time will be incremented until the first time $w$ is no longer served by that agent. This provides insight into the handoff that receiving vehicles experience. We see in Fig.~\ref{fig:2c} that three out of four receiving vehicles have an average sojourn time of about 20\%. Over a time horizon of $100$ ms, each receiving vehicle stays in communication with the same transmitting vehicle on average $20$ ms. This is due to the high mobility scenario considered. 

In order to increase the sojourn time of each receiving vehicle, it is necessary for our proposed DRL framework to undergo major changes: to develop a two-time-scale DQL algorithm that acts successively in two different time scales. That is, in the first time scale (of the order of $100$ ms--time horizon), the agents must choose their coverage communication ranges. Then, they must choose the slice selection decision, RB and power allocation in a shorter time-slice scale (on the order of $5$ ms). This can be done by implementing two DQNs (for each agent) that are connected in series. The study of this alternative is left for our future work.

\section{Conclusions}\label{sec:cl}
In this paper, we have developed an online distributed scheme to solve the problem of network slicing, coverage selection, resource block and power allocation in vehicular networks using non-orthogonal multiple access. We provided a mathematical programming formulation of the problem and an explicit NP-hardness analysis. We then modeled the problem as a multi-agent Markov decision process to develop a multi-agent deep reinforcement learning algorithm that is based on the well-known fingerprinting method known as deep Q learning (DQL). Our DQL algorithm utilized the most recent advances in deep reinforcement learning, including the duel architecture and prioritized experience replay memory. The proposed DQL algorithm has been shown to be robust and effective against various system parameters, including the high mobility characteristics of vehicular networks. It also outperformed baseline benchmark algorithms that are based on global and centralized decisions. In our future work, we will investigate the hierarchical two-time scale approach to solve the problem by proposing a different deep network architecture. In addition, we will consider the case of multiple slices rather than two and study the fairness between different slices.

\bibliographystyle{spmpsci}
\bibliography{IEEEabrv,M2Mv1}

\end{document}